\documentclass{llncs}
\usepackage{amsfonts,amssymb,amsmath}

\title{Regular realizability problems and  context-free languages}

\author{A. Rubtsov\thanks{Supported in part by RFBR grant
    14--01--00641.}\inst{2}\inst{3}
\and
M. Vyalyi\thanks{Supported  in part
RFBR grant 14--01--93107 and the scientific school grant %
NSh4652.2012.1.}
\inst{1}\inst{2}\inst{3}}
\institute{Computing Centre of RAS\\
\and Moscow Institute of Physics and Technology\\
\and National Research University Higher School of Economics\\
\email{rubtsov99@gmail.com}\\
\email{vyalyi@gmail.com}
}

\date{}
\spnewtheorem{Def}{Definition}{\bfseries}{\upshape}
\spnewtheorem{prop}{Proposition}{\bfseries}{\itshape}
\spnewtheorem{Claim}{Claim}{\bfseries}{\itshape}
\spnewtheorem*{known}{Theorem}{\bfseries\upshape}{\itshape}
\spnewtheorem*{Lemma}{Lemma}{\bfseries\upshape}{\itshape}
\let\eps\varepsilon
\let\al\alpha
\let\leq\leqslant

\let\ph\varphi

\let\es\varnothing

\def\ZZ{\mathbb Z}

\def\Lin{\ensuremath{\mathbf{Lin}}}
\def\Qrt{\ensuremath{\mathbf{Qrt}}}
\def\L{{\mathbf{L}}}

\def\NL{\ensuremath{\mathbf{NL}}}
\def\PP{\ensuremath{\mathbf{P}}}
\def\NP{\ensuremath{\mathbf{NP}}}
\def\PSPACE{\ensuremath{\mathbf{PSPACE}}}
\def\NSPACE{\ensuremath{\mathbf{NSPACE}}}

\def\CFL{\ensuremath{\mathbf{CFL}}}

\def\A{\ensuremath{\mathcal A}}
\def\C{\ensuremath{\mathcal C}}

\def\B{\ensuremath{\mathcal B}}

\def\T{\ensuremath{\mathcal T}}

\def\reg{\mathrm{RR}}
\def\nreg{\mathrm{NRR}}

\def\leGen#1#2{\mathbin{\leq^{\mathrm{#2}}_{\mathrm{#1}}}}

\def\lelog{\leGen{log}{}}

\def\lerat{\mathbin{\leq_{\mathrm{rat}}}}

\def\ba{\bar a}

\def\bx{\bar x}

\begin{document}

\maketitle

\begin{abstract} 
We investigate regular realizability (RR) problems, which are the problems
of verifying whether the intersection of a regular language -- the input of the
problem -- and a fixed language, called a filter, is non-empty. In this paper
we focus on the case of context-free filters. The algorithmic complexity
of the  RR problem  is a very coarse measure
of the complexity of context-free languages. 
This characteristic respects the rational dominance relation. 
We show that a RR problem for a maximal filter under the rational
dominance relation is \PP-complete. 
On the other hand, we present an example of a \PP-complete RR problem for a non-maximal filter.  We show that RR problems for Greibach languages belong to the class \NL. We also discuss RR problems with context-free filters that
might have intermediate complexity. 
Possible candidates are the languages with polynomially-bounded rational indices. We show that RR problems for these filters lie in the class $\NSPACE(\log^2 n)$.
\end{abstract}

\section{Introduction}

The context-free languages form one of  the most important classes for
formal language theory. There are many ways to characterize
complexity of context-free languages. In this paper we propose a new approach to
classification of context-free languages based on the algorithmic complexity of the
corresponding regular realizability (RR) problems.

By  `regular realizability' we mean the problem
of verifying whether the intersection of a regular language -- the input of the
problem -- and a fixed language, called a filter, is non-empty. The filter $F$ is a
parameter of the problem. Depending on the representation of a regular
language, we distinguish the deterministic RR problems $\reg(F)$ and
the nondeterministic ones $\nreg(F)$, which 
correspond to the description of the regular language either by a
deterministic or by a  nondeterministic finite automaton.

The relation between algorithmic complexities of $\reg(F)$ and
$\nreg(F)$ is still unknown. For our purpose -- the
characterization of the complexity of a context-free language -- the
nondeterministic version is more suitable. One of the reasons for this choice is a rational dominance relation $\lerat$ (defined in Section \ref{sec-preliminaries}).
We show below that the dominance relation on filters $F_1 \lerat F_2$ implies the log-space reduction $\nreg(F_1) \lelog \nreg(F_2)$.
So our classification is a very coarse version of the
well-known  classification of $\CFL$ by  
the  rational dominance  relation 
(see the book~\cite{Be09} for a~detailed exposition of this topic). 

Depending on a filter $F$, the algorithmic complexity of the regular realizability
problem varies drastically. There are RR problems that are complete for 
complexity classes such as 
$\L$, $\NL$, $\PP$, $\NP$, $\PSPACE$~\cite{ALRSS09,Vya11}. 
In~\cite{Vya13} a huge range of
possible algorithmic complexities of the deterministic RR problems was presented. 
We prove below
that for context-free nonempty filters the possible complexities are in the
range between \NL-complete problems and \PP-complete problems. Examples of $\PP$-complete RR problems are provided
in Section~\ref{hard}. The filter consisting of all words provides an easy example of an \NL-complete RR problem.
In this case, the problem is exactly the reachability
problem for digraphs.  The upper bound by the class $\PP$ follows from
the reduction of an arbitrary $\nreg$-problem specified by a context-free filter to the problem of
verifying the emptiness  of a language generated by a context-free
grammar. We prove it in Section~\ref{hard}.

We will call a context-free language $L$ \emph{easy} if $\nreg(L)\in\NL$ and
\emph{hard} if $\nreg(L)$ is \PP-complete. In Section~\ref{hard} we
present an example of a non-generator of the CFLs cone, which is hard
in this sense. In Section~\ref{easy} we provide examples of easy
languages. 
They cover a rather wide class -- the so-called Greibach languages
introduced 
in~\cite{Gre}.

The exact border between hard and easy languages is unknown. Moreover,
there are candidates for an intermediate complexity of RR
problems. 
They are languages with polynomially-bounded rational indices.

The rational index was
introduced in~\cite{BCN}. Recall that \emph{rational index} $\rho_L(n)$ of a language $L$ is 
a function that returns the maximum length of the 
shortest word from the intersection of the language $L$ and a language
$L(\A)$ recognized by an automaton $\A$ with $n$ states, provided
$L(\A)\cap L\ne\es$:
\begin{equation}\label{RatIndDef}
 \rho_L(n) = \max_{\A: |Q_\A| = n,\; L(\A)\cap L\ne\es}
\min\{ |u| \,|\, u \in L(\A) \cap L \}. 
\end{equation}
The growth rate of the language's rational index is an another measure of the complexity of a language. This measure is also related to the rational dominance (see Section~\ref{index} for details).

In Section~\ref{index} we prove that the RR problem for a context-free
filter having  polynomially-bounded rational index is in the class
$\NSPACE(\log^2 n)$. Note also that there are many known CFLs having 
polynomially-bounded rational indices~\cite{PF90}. 
But the RR problems for these languages
are in \NL. It would be interesting to find more sophisticated
examples of CFLs having
polynomially-bounded rational indices.

\section{Preliminaries}\label{sec-preliminaries}

The main point of our paper is investigation of the complexity of the
$\nreg$-problem for filters from the class of context-free languages  $\CFL$. 

\begin{Def}
	The regular realizability problem $\nreg(F)$ is the problem
	of verifying non-emptiness of the intersection of the filter $F$
with a regular language $L(\A)$, where $\A$ is an NFA. 
Formally
	$$\nreg(F) = \{ \A \mid \A \text{ is an NFA and } L(\A)\cap F
	\neq \es  \}. $$ 
\end{Def}

It follows from the definition that the problem $\nreg(A^*)$ for the
filter consisting of all words under alphabet $A$ is the well-known \NL-complete problem
of digraph reachability. We will show below that $\nreg(L)\in\PP$ for
an arbitrary context-free filter~$L$. So it is suitable to use deterministic
log-space reductions in the analysis of algorithmic complexity of the RR
problems specified by  CFL filters. We denote the deterministic
log-space reduction by $\lelog$.

Let us recall some basic notions and fix notation concerning the
CFLs. For a detailed exposition see~\cite{Be09,BeBoa90}.
We will refer to the empty word as $\eps$.
Let $A_n$ and $\bar A_n$ be the $n$-letter alphabets consisting of
the letters $\{a_1, a_2,\ldots,a_n\}$ and $\{\ba_1, \ba_2,\ldots,\ba_n\}$
respectively. 
A well-known example of a context-free language,
the \emph{Dyck language} $D_n$,  
is defined  by the grammar
$$  S
\to SS \mid \eps\, \mid a_1S\ba_1 \mid \cdots \mid
a_nS\ba_n. 
$$ 

Fix alphabets $A$ and $B$. A~language $L \subseteq A^* $ is \emph{rationally dominated}
by $L^\prime \subseteq B^* $ if there exists a rational relation
$R$ such that $L=R(L^\prime)$, where $R(X) =\{u \in A^*\mid \exists v \in X\ (v,u)\in R\}.$ 
We denote rational domination as $\lerat$.
We say that languages $L$, $L'$ are \emph{rationally equivalent} if
$L\lerat L'$ and $L'\lerat L$.   

A rational relation is a graph of a multivalued mapping $\tau_R$. We
will call the mapping $\tau_R$  with a rational graph
as a rational transduction. So $L\lerat L'$ means that $L = \tau_R(L^\prime)$.
Such a transduction can be realized by a \emph{rational transducer} 
(or finite-state transducer) $T$, which is a
nondeterministic finite automaton with input and output tapes, where
$\eps$-moves are permitted. We say that $u $ belongs to $ T(v)$ if
for the input $v$ there exists a path of  computation on which $T$ 
writes the word $u$ on the output tape and halts in the accepting
state. 
Formally, a  rational transducer is defined by the 6-tuple 
$T = (A, B, Q, q_0, \delta, F)$, where 
$A$ is the input alphabet, $B$ is the output alphabet, 
$Q$ is the (finite) state set, $q_0$ is the initial
state,  $F \subseteq Q$ is the set of accepting states and 
$\delta\colon Q \times (A\cup\eps) \times
(B\cup\eps)\times Q $ is the
transition
relation. 
 
Let two rational transducers $T_1$ and $T_2$ correspond to
rational relations $R_1$ and $R_2$, respectively. We say that a rational transducer 
$T= T_1 \circ T_2$ is the composition of $T_1$ and $T_2$ if the relation $R$
corresponding to $T$ such that $R =  \{ (u,v)\mid\exists y (u,y) \in R_1,
(y,v) \in R_2 \}$.

Define the composition of transducer $T$ and automaton $\A$ in the
same way: automaton $\B = T \circ \A$ recognizes the language $\{w \,|\, \exists y \in L(\A)\ (w,y) \in R \}$.

The following proposition is an algorithmic version of the Elgot-Mezei
theorem (see, e.g., \cite[Th. 4.4]{Be09}).

\begin{prop}\label{compose}
 The composition of transducers and the composition of a transducer
   and an automaton are computable in deterministic log space.
\end{prop}

A \emph{rational cone} is a class of languages closed under rational dominance.
Let $\T(L)$ denote the least rational cone that includes language $L$
and call it the \emph{rational cone generated by $L$}. 
Such a cone is called \emph{principal\/}.
For example, the cone $\Lin$ of linear languages (see~\cite{Be09} for
definition) is principal: $\Lin =\T(S)$, where  the \emph{symmetric
language} $S$ over the 
alphabet $X = \{x_1,x_2,\bx_1,\bx_2\}$ is defined by the grammar
\begin{equation*}
     S\to x_1 S\bx_1 \mid x_2 S\bx_2 \mid \eps.
\end{equation*}

For a mapping $a\mapsto L_a$ the \emph{substitution} $\sigma$ is the morphism
from $A^*$ to the power set $2^{B^*}$ such that $\sigma(a) = L_a$. The
image $\sigma(L)$ of a language $L\subseteq A^*$ is defined in the natural
way. The \emph{substitution closure} of a class of languages $\mathcal
L$ is the least class containing all substitutions of languages from
$\mathcal L$ to the languages from~$\mathcal L$. We need two
well-known examples of the substitution closure. The class
$\Qrt$ of the \emph{quasirational languages} is the substitution closure of
the class $\Lin$. The class of \emph{Greibach languages}~\cite{Gre} is the
substitution closure of the rational cone generated by the Dyck language $D_1$ and the symmetric
language~$S$.

It is important for our purposes that rational dominance implies a
reduction for the corresponding RR problems.

\begin{lemma}\label{transductions-vs-logspace-reductions}
	If $F_1 \lerat F_2$ then $\nreg(F_1) \lelog \nreg(F_2) $.
\end{lemma}
\begin{proof}
	Let $T$ be a rational transducer such that $F_1 = T(F_2)$
	and let $\A$ be an input of the $\nreg(F_1)$ problem. Construct the
	automaton $\B = T \circ \A $ and use it as an input of the
	$\nreg(F_2)$ problem. 
	It gives the log-space reduction due to
	Proposition~\ref{compose}.  
\end{proof}

In particular, this lemma implies that  if a problem $\nreg(F)$ is
complete in a complexity class $\C$, then for any filter $F^\prime$ from the
rational cone $\T(F)$  the problem $\nreg(F^\prime)$  is in the class~$\C$. 

We will use the following reformulation of the Chomsky-Sch\"utzenberger theorem.

\begin{known}[Chomsky, Sch\"utzenberger]
	$\CFL = \T(D_2)$.
\end{known}
In the next section, we prove  that $\nreg(D_2)$ is  $\PP$-complete
under deterministic log-space reductions. Thus, it 
follows from the Chomsky-Sch\"utzenberger theorem and
Lemma~\ref{transductions-vs-logspace-reductions} that any problem $\nreg(F)$
for a CFL filter $F$ lies in the class $\PP$.

\section{Hard  RR problems with CFL filters}\label{hard}

In this section we present examples of hard context-free
languages. The first example is the Dyck language
$D_2$. 

By use of
Lemma~\ref{transductions-vs-logspace-reductions} and
the Chomsky-Sch\"utzenberger theorem, we conclude that
any generator of the CFL cone is hard. But there are additional hard
languages. We provide such an example, too.

We start with some technical lemmas. The intersection of a CFL and a
rational language is a CFL. 
We need an algorithmic version of this fact.

\begin{lemma}\label{ATrans}
	Let $G = (N, \Sigma, P, S)$ be a fixed context-free grammar.
	Then there exists a deterministic log-space algorithm that 
	takes a description of an NFA $\A
	= (Q_\A, \Sigma, \delta_\A, q_0, F_\A)$  and
	constructs a
	grammar $G^\prime = 
	 (N^\prime, \Sigma, P^\prime, S^\prime) $ generating the language
	 $L(G)\cap L(\A)$. The grammar size is  polynomial  in
	 $|Q_\A|$. 
\end{lemma}

This fact is well-known. We provide the proof because the construction
will be used  in the proof of Theorem \ref{PolynomialRatIntdex} below.

\begin{proof}[of Lemma~\ref{ATrans}]
	First, to make the construction clearer, we assume that automaton $\A$ has no $\eps$-transitions.
	Let $N^\prime$ consist of the
	axiom $S^\prime$ and nonterminals $[qAp]$, where $A \in N$ and
	$q,p \in Q_\A$. Construct $P^\prime$ by adding for each rule
	$A \to X_1X_2\cdots X_n $ from $P$ the set of rules 
	$$\{ [qAp]\to [qX_1r_1][r_1X_2r_2]\cdots[r_{n-1}X_np]\mid
	q,p,r_1,r_2,\ldots,r_{n-1} \in Q_\A \}
	$$
	to
	$P^\prime$. Also add to $P^\prime$ rules $[q\sigma p]\to
	\sigma $ if $\delta_\A(q,\sigma) = p$ and $S^\prime \to
	[q_0Sq_f]$ for each $q_f$ from $F_\A$.
	
	Now we prove that $L(G^\prime) = L(G) \cap L(\A)$. Let $G$
	derive the word $w = w_1w_2\cdots w_n$.  
Then grammar $G^\prime$ derives all possible sentential forms
$$[q_0w_1r_1][r_1w_2r_2]\cdots [r_{n-1}w_nq_f],$$ 
where $q_f \in F_\A$
and $r_i \in Q_\A$. And $[q_0w_1r_1][r_1w_2r_2]\cdots [r_{n-1}w_nq_f]
\Rightarrow^* w_1w_2\cdots w_n $ iff there is a successful run for
the automaton $\A$ on $w$. If $G^\prime$ derives a word $w$ then each
symbol $w_i$ of the word has been derived from some nonterminal
$[qw_ip]$.  Due to the construction of the grammar $G^\prime$ the word $w$
has been derived from some sentential form
$[q_0w_1r_1][r_1w_2r_2]\cdots [r_{n-1}w_nq_f],$
which encodes a successive run of $\A$ on $w$. Thus 
$G^\prime$ derives the word $w$ only if $G$ does as well.

The size of $G^\prime$ is polynomial in $Q_\A$. The size of $N^\prime$
is $|N|\cdot|Q_\A|^2+1$. Let $k$ be the length of the longest rule in
$P$. Then for each rule from $P$ there  are at most $|Q_\A|^{k+1}$
rules in $P^\prime$ and for rules 
in the form $[q\sigma p] \to \sigma$ or
$S^\prime \to [q_0Sq_f]$ 
there are at most
$O(|Q_\A|^2)$ rules in $P^\prime$.

Finally, the grammar $G^\prime$ is log-space constructible, because 
the rules of $P^\prime$ corresponding to the particular rule from $P$ can be
generated by inspecting all $(k+1)$-tuples of states of $\A$ and $k=O(1)$.
Adding $\eps$-transitions just increases $k+1$
to $2k$. For each rule $A \to X_1\cdots X_n$ we add rules
$[qAp] \to [qX_1q_1][q_2X_2q_3]\cdots [q_{2n-1}X_np]$, where 
$q_i=q_{i+1}$ or $q_i \xrightarrow{\eps} q_{i+1}$ for all~$i$.
In the case of $[q\sigma p]\to \sigma$ 
rules we add all such rules that $q
\xrightarrow{\eps} q^\prime $, $p^\prime \xrightarrow{\eps} p$ and
$\delta(q^\prime, \sigma) = p^\prime$.
\end{proof}

Note that if grammar $G$ is in Chomsky normal form, then the number of nonterminals of the grammar $G^\prime$ is $O(|Q_\A|^2)$.  Recall that for a grammar in the Chomsky normal form, the right-hand
side of each rule consists of either two nonterminals, or one
terminal. The empty word may be produced only by the axiom and the axiom does not appear in a right-hand side  of any rule.

Also we need an algorithmic version of the Chomsky-Sch\"utzenberger theorem.

\begin{lemma}\label{CS-alg}
  There exists a deterministic log-space algorithm that takes  a
  description of a context-free grammar $G = (N,\Sigma, P, S)$ and
  produces a rational transducer $T$ such that  $T(D_2) = L(G)$.
\end{lemma}

Now we are ready to prove hardness of the Dyck language $D_2$.

\begin{theorem}
 The problem $\nreg(D_2)$ is $\PP$-complete.
\end{theorem}
\begin{proof}
To prove \PP-hardness 
we reduce the well-known $\PP$-complete problem of
verifying whether a context-free grammar generates an empty language~\cite{GHR}
to $\nreg(D_2)$. 
Based on a grammar $G$, construct a transducer $T$ such that 
$T(D_2) =L(G)$ 
using Lemma~\ref{CS-alg}.  Let $\A$ be a
nondeterministic automaton obtained from the transducer $T$ by
ignoring the output tape. Then $L(\A)\cap D_2$ is nonempty iff
$L(G)$ is nonempty. The mapping $G\to \A$ is the required reduction.

To prove that $\nreg(D_2)$ lies in $\PP$ we reduce this problem to
the problem of non-emptiness of a language generated by a context-free
grammar. 

For an input $\A$ construct the grammar $G$ such that  
$L(G) = L(\A)\cap D_2$ using Lemma~\ref{ATrans}. 
\end{proof}

\begin{corollary}
  Any generator of the \CFL{} cone is a hard language. 
\end{corollary}

Now we present another example of a hard language.  Boasson proved
in~\cite{Boa85} that there exists a principal rational cone of
non-generators of the CFL cone containing the family $\Qrt$ of the quasirational languages.

Below we establish \PP-completeness of the nondeterministic RR
problem for a generator of this cone. The construction follows the
exposition in~\cite{BeBoa90}.

For brevity we denote the alphabet of the Dyck language $D_1$ by  $A = \{a,\ba\}^* $.
Recall that the syntactic substitution of a language $M$ into a language $L$ is
$$ L\mathbin{\uparrow}M = \{ m_1x_1m_2x_2\cdots m_r x_r \mid m_1,\ldots, m_r \in M,\, x_1x_2\cdots x_r\in L  \}\cup (\{\eps\}\cap L). $$
We also  use the language $S_\# = S\mathbin{\uparrow} \#^*$ which
is the \emph{syntactic substitution} of the language $\#^*$ in the symmetric
language~$S$.

Let $M = a S_\# \ba \cup \eps$.
The language $M^{(\infty)}$ is defined recursively in the following way:
$x\in M^{(\infty)}$ iff either $x\in M$ or
$$
x = ay_1a z_1\ba y_2az_2\ba\cdots y_{n-1} az_{n-1}\ba y_n\ba,
$$
where $y_1,y_n\in X^*$, $y_i\in X^+$ for $2\leq i\leq n-1$,
$az_i\ba\in M^{(\infty)}$ and $ay_1y_2\cdots y_n\ba\in M$.

Let $\pi_X\colon (X\cup A)^*\to A^*$ 
be the morphism that erases symbols
from the alphabet~$X$. The language $M^{(+)}$ is defined to be
$\pi_X^{-1}(A^*\setminus D_1)$. 

Finally, we set $S_\#^\uparrow = M^{(\infty)}\cup M^{(+)}$. 

Note that the languages $S$ and $S_\#$ are rationally equivalent. 
So $S_\#$  is a generator of the cone $\Lin$
of the linear languages.

By combining this observation with Propositions 3.19 and 3.20
from~\cite{BeBoa90}, we get the following fact.

\begin{theorem}
  $S_\#^\uparrow$ is not a generator of the \CFL{} cone, but the cone
  generated by $S_\#^\uparrow$ contains all quasirational languages.
\end{theorem}

The language $S_\#^\uparrow$ is the union of two languages. 
In the proof of the \PP-completeness  for the problem
$\nreg(S_\#^\uparrow)$, we will use automata that do not accept words
from the language~$M^{(+)}$. 
For this purpose we need a notion of a marked automaton.

\begin{Def}
  An NFA $\A$ over the alphabet $A_n\cup \bar A_n$ 
  is \emph{marked} if there exists
  a function $h\colon 
  Q_\A\to \ZZ$ satisfying the relations
  \begin{equation}
    \begin{aligned}
      &h(q') =  h(q)+1, &&\text{if there exists a transition }
      q\xrightarrow{a_j}q'\ \text{in }\A,\\
      &h(q') =  h(q)-1, &&\text{if there exists a transition }
      q\xrightarrow{\ba_j}q'\ \text{in }\A,\\
      &h(q)=0, &&\text{if  $q$ is either the initial state or an
      accepting state of $\A$.}\vphantom{q\xrightarrow{a_j}q'}
    \end{aligned}
  \end{equation}
\end{Def}

In what follows we will identify for brevity the (directed)
paths along the graph of an NFA and the corresponding words in the
alphabet of the automaton. The vertices of the graph, i.e., the states
of the automaton, are identified in this way with the
\emph{positions} of the word.

The \emph{height} of a position  is the
difference between the number of the symbols~$a_i$
and the number of the symbols~$\ba_i$   preceding the position.
In terms of the position heights,
the words in $D_1$ are characterized by two conditions: the height
of any position is nonnegative and the height of the final position is~0.

\begin{prop}\label{height-pos}
  Let $\A$ be an NFA such that $D_2\cap L(\A)\ne\es$. Then there exists
  a word $w\in D_2\cap L(\A)\ne\es$ such that the height of any
  position in the word $w$ is $O(|Q_\A|)^2$.
\end{prop}
\begin{proof}
  The heights of positions are upperbounded by the height of
  the derivation tree in the grammar generating the language $D_2\cap
  L(A)\ne\es$.

  It is easy to see that for any grammar generating a non-empty language
  there is a word such that the height of a derivation tree for the
  word is at most the number of nonterminals in the grammar. 

  To finish the proof, we use the grammar constructed by Lemma~\ref{ATrans}
from the grammar generating $D_2$ in the Chomsky normal form. This grammar has $O(|Q_\A|^2)$ nonterminals.
\end{proof}

In the proof below we need a syntactic transformation of automata
over the alphabet $A_2\cup \bar A_2$.

\begin{prop}\label{marking-transform}
  There exists a transformation $\mu$ that takes a description of an
  automaton $\A$ over the alphabet $A_2\cup \bar A_2$ and produces a
  description of a marked automaton $\A'=\mu(\A)$ such that 
  (i) $L(\A)\cap  D_2\ne\es$ iff $L(\A^\prime)\cap  D_2\ne\es$
  and (ii) for any $ w\in L(\A^\prime)$ the height of any position is
  nonnegative and the height of the final position is~$0$.
  The   transformation $\mu$ is computed in deterministic log space.
\end{prop}
\begin{proof}
  Let $m$ be an upper bound on the heights of the positions in a word
  $w\in L(\A)\cap  D_2$. By
  Proposition~\ref{height-pos}, $m$ is $O(|Q_\A|^2)$ . Note that $m$ can be computed in
  deterministic log space.

  The state set of the automaton $\A'$ is $Q_\A\times\{0,\dots, m\}\cup\{r\}$,
  where $r$ is the specific absorbing rejecting state.

  If $q\xrightarrow{\al} q'$, where $\al\in \{a_1,a_2\}$, is a transition
  in the automaton $\A$ then there are transitions
  $(q,i)\xrightarrow{\al}(q',i+1)$ for all $0\leq i<m$ and the
  transition $(q,m)\xrightarrow{\al} r$ in the automaton $\A'$.

  If $q\xrightarrow{\al} q'$, where $\al\in \{\ba_1,\ba_2\}$, is a transition
  in the automaton $\A$ then there are transitions
  $(q,i)\xrightarrow{\al}(q',i-1)$ for all $0< i\leq m$ and the
  transition $(q,0)\xrightarrow{\al} r$  in the automaton $\A'$.
  
  The initial state of the automaton $\A'$ is $(q_0,0)$, where $q_0$ is
  the initial state of the automaton $\A$. The set of accepting states of the
  automaton $\A'$ is $F\times\{0\}$, where $F$ is the set of
  accepting states of the automaton $\A$.

  It is clear that the description of the automaton $\A'$ is
  constructed in deterministic log space.

  Condition (ii) is forced by the construction of the
  automaton~$\A'$. It remains to prove that condition (i) holds.

  Note that if $L(\A)\cap D_2=\es$ then $L(\A')\cap D_2=\es$ too. In the
  other direction, if $L(\A)\cap D_2\ne\es$, then by
  Proposition~\ref{height-pos} there exists a word $w\in L(\A)\cap D_2$
  such that the height of any position in the word does not
  exceed~$m$. So the word is accepted by the automaton~$\A'$.  
\end{proof}

\begin{theorem}
  $\nreg(S_\#^\uparrow)$ is $\PP$-complete under deterministic log
  space reductions.
\end{theorem}
\begin{proof}
  We reduce $\nreg(D_2)$ to $\nreg(S_\#^\uparrow)$.

  Let $\A$ be an input of the problem $\nreg(D_2)$ and $\A' = \mu(\A)$ be
  the marking transformation of the automaton~$\A$. 
  
  We are going  to construct the automaton $\B$ over the alphabet
  $A\cup X\cup\{\#\}$ such that $L(\A')\cap D_2\ne\es$ iff $L(\B)\cap
  S_\#^\uparrow\ne\es $. 

  The morphism $\ph\colon (A_2\cup\bar A_2)^*\to 
  (A\cup X\cup\{\#\})^* $ is defined as follows:
  \begin{equation}\label{D2->M}
    \begin{split}
      &\ph\colon a_1 \mapsto ax_1,\\
      &\ph\colon \ba_1\mapsto \bx_1\ba\#\#,\\
      &\ph\colon a_2 \mapsto ax_2,\\
      &\ph\colon \ba_2\mapsto \bx_2\ba\#\#.
    \end{split}
  \end{equation}

  The automaton $\B$ accepts words of the form $ax_1x_2 w
  \bx_2\bx_1\ba$, where $w = \ph(u)$. It simulates the
  behavior of the automaton $\A'$ on the word $u$ and accepts iff $\A'$
  accepts the word~$u$.

  It follows from the definitions that if $u\in D_2$ then $ax_1x_2
  \ph(u)\bx_2\bx_1\ba \in M^{(\infty)}$. So if $L(\A')\cap D_2\ne \es$
  then $L(\B)\cap S_\#^\uparrow\ne\es$. 

  Now we are going to prove the opposite implication. Let 
  $$
  w = ax_1x_2   \ph(u)\bx_2\bx_1\ba\in S_\#^\uparrow\cap L(\B).
  $$

   The automaton $\A'$ is marked and $\B$ simulates the behavior of
   $\A'$ on $u$. 
   So the heights of positions in $w$ are nonnegative and the height
   of the final position is~$0$. Thus $w\notin M^{(+)}=
   \pi_X^{-1}(A^*\setminus D_1)$.
   Take a pair of the
  corresponding parentheses $a$, $\ba$ in the word~$w$:
  $$
  w = w_0 ax_i w_1 \bx_j\ba w_2.
  $$
  If $i\ne j$ then $w\notin M^{(\infty)}$. So $i=j$ for all pairs of
  the corresponding parentheses. This implies $u\in D_2\cap L(\A')$.

  We just have proved the correctness of the reduction. 
It can be computed in log space due to the following observations.  
To produce the automaton $\B$ from the
  automaton $\A$ we need to extend the state set by a finite number of
  pre- and postprocessing  states to operate with the prefix $ax_1x_2$
  and with the suffix $\bx_2\bx_1\ba$. Also we need to split all
  states in $Q_{\A'}$ 
  in pairs to organize the simulation of $\A'$ while reading the pairs
  of symbols $ax_i$ and $\bx_i\ba$. The transitions by the symbol $\#$
  are trivial: $q\xrightarrow{\#}q$ for all~$q$. 
\end{proof}

\section{Easy  RR problems with CFL filters}\label{easy}

Now we present  examples of easy languages. The simplest example is
rational languages. Next we prove that the symmetric language and the
language $D_1$ are easy. A simple observation shows that
a substitution of easy languages into an easy language is easy. Thus
we conclude that Greibach languages are easy.

\begin{lemma}\label{Seasy}
  $\nreg(S)\in\NL$.
\end{lemma}

The proof of Lemma~\ref{Seasy} is a slight modification of the
arguments from~\cite{ALRSS09} that prove a similar result for the
language of palindromes.

\begin{lemma}\label{countertheorem}
	Let $L_c$ be a context-free language recognizable by a counter
	automaton. Then problem $\nreg(L_c)$ lies in $\NL$. 
\end{lemma}

In the proof we will use the following fact.

\begin{lemma}[\cite{Abuzer12}]\label{counterlemma}
  Let $M$ be a counter automaton with $n$ states. Then the shortest
  word $w$ from the language $L(M)$ has length at most $n^3$ and
  the counter 
  of $M$ on processing the word $w$ doesn't exceed the value $n^2$. 
\end{lemma}

We now return to the proof of Lemma~\ref{countertheorem}.

\begin{proof}
Let $M$ be a counter automaton that accepts  by reaching the final
state such that $M$  recognizes the language
$L_c$.   Let $\A$ be an automaton on the input of the regular
realizability 
problem. 

Construct the counter automaton $M_\A$ with the set of states $Q_M
\times Q_\A$, the initial state $(q^{M}_0,q^{\A}_0)$, with the set of
accepting states $F_M\times F_\A$ and with the transition relation
$\delta_{M_\A}$ such that $\delta_M(q,\sigma,z) \vdash
(q^\prime,z^\prime) $, $\delta_\A(p,\sigma) = p^\prime$ implies
$\delta_{M_\A}((q,p),\sigma, z) \vdash ((q^\prime,p^\prime),
z^\prime)$. This is the standard composition construction.

The automaton $M_\A$ is a counter automaton with
$|Q_M|\cdot|Q_\A|=c\times n$ states. Using Lemma \ref{counterlemma} we
obtain that the value of 
$M_\A$'s counter does not exceed $(cn)^2$ on the shortest word from
$L(M_\A)$. Then construct automaton $\B$ such that $L(\B)$ contains
all such words from $L(M_\A)$ such that the counter of $M_\A$ does not exceed $(cn)^2$. 
The automaton $\B$ has $O(n^3)$ states and can be constructed in log space in
the straightforward way similar to the proof of
Proposition~\ref{marking-transform}. 
Note that  $L(M_\A) \neq \es$ iff $L(\B) \neq \es$. So the map $\A\to
\B$ gives a reduction of the
problem $\nreg(L_c)$ to the problem $\nreg(\Sigma^*)$, which is in
$\NL$.
\end{proof}

The language $D_1$ is recognized by a counter automaton in the obvious
way. 

\begin{corollary}\label{D1easy}
  $\nreg(D_1)\in\NL$.
\end{corollary}

\begin{lemma}\label{substitution}
  If $L$, $L_a$ for all $a\in A$, are easy languages then $\sigma(L)$
  is also easy.
\end{lemma}
\begin{proof}
  Let $\A$ be an input for the problem $\nreg(\sigma(L))$. Define the
  automaton $\A^\prime$ over the alphabet $A$ with the state set
  $Q_{\A^\prime} = Q_\A$.  There is a transition $q\xrightarrow{a}q'$
  in the automaton $\A^\prime$ iff there exists a word $w\in L_a$ such
  that $q\xrightarrow{w}q'$ in automaton~$\A$.

  It is clear from the definition that $L(\A)\cap \sigma(L)\ne\es$ iff
  $L(\A^\prime)\cap L\ne\es$. To apply an $\NL$-algorithm for
  $\nreg(L)$ one needs the transition relation of $\A^\prime$. The transition relation is not a part of the
  input now. But it can be computed by $\NL$-algorithms for
  $\nreg(L_a)$. It is clear that the resulting algorithm is in $\NL$.
\end{proof}

Applying Lemma~\ref{substitution}, Lemma~\ref{Seasy} and
Corollary~\ref{D1easy}, we deduce with the theorem.

\begin{theorem}
  Greibach languages are easy.
\end{theorem}

\section{The case of polynomially-bounded rational index}\label{index}

We do not know whether there exists a CFL that is neither hard nor easy. 
In this section we indicate one possible class of candidates for an
intermediate complexity: the languages 
with polynomially-bounded rational indices.

Rational index appears to be a very useful
characteristic of a context-free language because 
rational index does not increase significantly under
rational transductions.

\begin{known}[Boasson, Courcelle, Nivat, 1981, \cite{BCN}]
	If $L^\prime \lerat L$ then there exists a constant $c$ such that
	$ \rho_{L^\prime}(n) \leq cn(\rho_{L}(cn)+1)$. 
\end{known} 

Thus the rational index can be used to separate languages
w.r.t. the rational dominance relation.
 Note  that the
rational index of a generator of the  $\CFL$ cone  has rather good estimations.

\begin{known}[Pierre, 1992, \cite{P92}]
	The rational index of any generator of the rational cone of $\CFL$
   belongs to $\exp(\Theta(n^2/\log n))$. 
 \end{known}

The examples of easy languages in Section~\ref{easy} 
have polynomially-bounded rational indices. 
Moreover, context-free languages with rational index $\Theta(n^\gamma)$ for any positive algebraic number
$\gamma>1 $ were presented in~\cite{PF90}. 
All of them are easy. The proof is rather technical and is skipped here. 
Thus it is quite natural to suggest that any
language with polynomially-bounded rational index is easy. 

Unfortunately we are able to give only a weaker bound on the algorithmic
complexity in the case of polynomially-bounded rational index.

\begin{theorem}\label{PolynomialRatIntdex}
  For a context-free filter $F$
with polynomially-bounded rational index,
  the problem $\nreg(F)$ lies in $\NSPACE(\log^2 n)$.
\end{theorem}

We use a technique quite similar to the technique from
\cite{MemBounds}. First we need an auxiliary result.

\begin{Lemma}[\cite{MemBounds}]
  For a grammar $G$ in the Chomsky normal form and for an arbitrary string
  $w = xyz$ from $L(G)$ of length $n$ there is a nonterminal $A$ in
  the derivation tree, such that $A$ derives $y$ and $n/3\leq|y|\leq 2n/3$.
\end{Lemma}

Let us return to the proof of the theorem.

\begin{proof}[of Theorem~\ref{PolynomialRatIntdex}]
  Consider a grammar $G^\prime$ in the Chomsky normal form such that
  $L(G^\prime) = F$. Fix an automaton $\A$ with $n$ states such that
  the  minimal length of $w$ from $L(\A)\cap F$ equals
  $\rho_F(n)$. 
  The length of the word $w$ is polynomial in $n$. Consider
  the grammar $G$ such that $L(G) = L(\A)\cap F$ 
obtained from the grammar
  $G^\prime$ by the construction from Lemma~\ref{ATrans}.

  The algorithm does not construct the grammar $G$ itself, 
  since such a construction 
  expands the size of grammar $G^\prime$ up to $n^3$
  times. Instead, the algorithm
  nondeterministically guesses the derivation tree of the word $w$ in the grammar
  $G$, if it exists.
  Informally speaking, it restores the derivation tree
  starting from its `central' branch.

  The main part of the algorithm is a recursive procedure that checks
  correctness for a nonterminal $A = [qA^\prime p]$ of the
  grammar~$G$.  We say that the nonterminal $A = [qA^\prime p]$ is
  correct if $A$ produces a word $w$ in the grammar~$G$.

  If a nonterminal is  $[q\sigma p]$, where $\sigma$ is a terminal
  then the procedure
  should check that  $q\xrightarrow{\sigma}p$ in
  the automaton~$\A$. 

  In a general case the procedure of checking correctness 
  nondeterministically guesses a 
  nonterminal $A_1=[\ell_1A_1'r_1]$ such 
  that $w=p_1u_1s_1$, and  $A_1$
  derives the word $u_1$ and $ 1/3|w| \leq |u_1| \leq 2/3|w| $.
  Then it is recursively applied to the nonterminal $A_1$.
  If successful the procedure sets $i:=1$ 
  and repeats the following steps:
  \begin{enumerate}
  \item Nondeterministically guess the ancestor $A_{i+1} =
  [\ell_{i+1}A_{i+1}r_{i+1}]$ of $A_i$ in 
  the derivation tree. There are two possible cases:
  \begin{itemize}
  \item [(i)] either $A_{i+1}\to [q'C'\ell_{i+1}]A_i$ in the grammar $G$ 
    (set up $C:=[q'C'\ell_{i+1}] $)
  \item [(ii)] or $A_{i+1}\to A_i[r_{i+1}C'p']$
    (set up $C:=[r_{i+1}C'p'] $).
  \end{itemize}
  \item Recursively apply the procedure of checking correctness to the nonterminal
  $C$. 
  \item If successful set up $i:=i+1$.
  \end{enumerate}

  Repetitions are finished and the procedure returns success if $A_j = A$. 
  If any call of the procedure of checking correctness returns failure then the whole
  procedure returns failure.

  In recursive calls the lengths of words to be
  checked diminish by a factor at most $2/3$. So the total number of
  recursive calls is $O(\log n)$, where $n$ is the input length. 
  Data to be stored during the process form a  list of triples
  (an automaton state, a nonterminal of the grammar $G'$,  a automaton
  state). Each automaton state description requires $O(\log n)$ space
  and nonterminal description requires a constant size space since 
  grammar $G^\prime$ is fixed. Thus
  the total space for the algorithm is $O(\log^2 n)$.	
\end{proof}

\subsection*{Acknowledgments}

We are acknowledged to Abuzer Yakaryilmaz for pointing on the result
of Lemma~\ref{countertheorem} and for reference to a lemma similar to
Lemma~\ref{counterlemma}.

\end{document}